\documentclass[11pt]{article}
\usepackage[T1]{fontenc}
\usepackage{lmodern,microtype}
\usepackage[margin=0.95in]{geometry}
\usepackage{amsmath,amssymb,amsthm,mathtools}
\usepackage{booktabs,tabularx,array}
\usepackage[dvipsnames]{xcolor}
\usepackage{enumitem,needspace,float}
\usepackage{caption,fancyhdr,titlesec}
\usepackage{hyperref}
\definecolor{Ink}{HTML}{17324D}
\definecolor{Teal}{HTML}{176B78}
\definecolor{Gray}{HTML}{667680}
\definecolor{Rule}{HTML}{D8E0E4}
\hypersetup{colorlinks=true,linkcolor=Teal,citecolor=Teal,urlcolor=Teal,
 pdftitle={Removing m from m-pile Divisor Nim},
 pdfauthor={Satyam Tyagi}}
\newtheorem{theorem}{Theorem}[section]
\newtheorem{lemma}[theorem]{Lemma}
\newtheorem{proposition}[theorem]{Proposition}
\newtheorem{corollary}[theorem]{Corollary}
\DeclareMathOperator{\mex}{mex}

\newcommand{\join}{\sqcup}
\newcolumntype{Y}{>{\raggedright\arraybackslash}X}
\newcolumntype{L}[1]{>{\raggedright\arraybackslash}p{#1}}

\titleformat{\section}{\Large\bfseries\color{Ink}}{\thesection}{0.6em}{}
\titleformat{\subsection}{\large\bfseries\color{Teal}}{}{0pt}{}
\titlespacing*{\section}{0pt}{15pt}{7pt}
\titlespacing*{\subsection}{0pt}{10pt}{5pt}
\setlist[itemize]{leftmargin=1.4em,itemsep=2pt,topsep=3pt}
\renewcommand{\headrulewidth}{0.4pt}
\renewcommand{\headrule}{\hbox to\headwidth{\color{Rule}\leaders\hrule height\headrulewidth\hfill}}
\allowdisplaybreaks[1]
\begin{document}
\begin{center}
 {\fontsize{23}{27}\selectfont\bfseries\color{Ink}
 Removing $m$ from $m$-pile Divisor Nim\par}
 \vspace{9pt}
 Satyam Tyagi
\end{center}
\vspace{2pt}

\begin{abstract}
In Divisor Nim, a player removes from one heap a positive number of stones that
divides the size of every other heap. We prove that the Sprague--Grundy (SG) value
satisfies $g(P)\le 2h_{\min}+\lfloor\log_2h_{\min}\rfloor$, where $h_{\min}$ is the
smallest nonempty heap, independently of the number and sizes of the other heaps.
The proof starts with odd heaps and classifies removals by their factors of $2$.
The largest heap in a starting board then supplies an SG bound valid throughout
every successor chain.

When one heap varies, we sharpen Morisawa's eventual-period bound: the smallest
fixed heap controls the SG value range, the number of consecutive values compared,
and an increase in the varying heap that repeats legal moves. The remaining
dependence on the fixed board lies in the SG periods of direct successors. Using
the largest fixed heap, we derive an explicit common period for all these recursive
layers. Its possible prime divisors belong to a fixed finite set, independent of
heap count; their exponents remain controlled by the total fixed stones.
Finally, joining boards under one shared divisor rule preserves and can tighten
the SG ceiling. Minimum 2-adic depth and its count parity determine every joined
outcome, and a join with at least two odd heaps has exact SG value $0$ or $1$.

We close with two conjectures: $g(P)\le2h_{\min}$, and that the least eventual SG
period, when one heap varies, divides $2\operatorname{lcm}(1,\ldots,h_{\max})$.
Here $h_{\max}$ is the largest fixed heap, excluding the varying heap. The second
claim would remove heap-count dependence from the period bound as well.
\end{abstract}

\section{Can one small heap control a large board?}
\label{sec:intro}

A board contains a heap of size $1$. We add many more heaps, each as
large as we wish. Can the Grundy value grow without limit? In Divisor
Nim, it cannot. While the heap $1$ remains unchanged, every removal
from another heap must divide $1$ and therefore must equal $1$.
We will show that the value stays at most $2$, then extend the argument
to any smallest heap.

The same idea leads to three questions. How large can the Grundy
value be? How long can a period of these values be when one heap
varies? What information survives when boards are joined? We address
them in that order.

\subsection{The rules and the value being bounded}

A position $P=(h_1,\ldots,h_m)$ lists the positive heap sizes ($h_i \ge 1$).
A move subtracts a positive integer $d$ from one heap $h_i$, subject to
\begin{equation}
 d\le h_i,\qquad d\mid h_j\quad\text{for every }j\ne i.
 \label{eq:rule}
\end{equation}
Here $d\mid h_j$ means that $d$ divides $h_j$. The removal need not
divide the heap being changed. Empty heaps impose no restriction,
since every positive integer divides $0$; we omit them from the displayed
board. Thus $m$ denotes the current number of nonempty heaps.
If only one heap remains then every removal up to its size is permitted.
Players alternate, and the player with no legal move loses.
These are Morisawa's rules for $m$-pile Divisor Nim
\cite{Morisawa2026}.

A \emph{successor} is a board left after one legal move.
The \emph{Sprague--Grundy (SG) value} of a board is defined recursively
\cite{Siegel2013}. The empty board has value $0$. For another board, list the values of its successors and choose the smallest nonnegative integer missing from the list:
\begin{equation}
 g(P)=\mex\{g(Q):P\longrightarrow Q\}.
 \label{eq:mex}
\end{equation}
The notation $\mex$ means the smallest missing nonnegative integer.
A board is losing exactly when its SG value is $0$, and winning
when its SG value is positive. A sole heap of size $h$ has
SG value $h$.

The distinction between moves and values is essential here.
A thousand moves to boards with SG value $0$ still supply only
one value to the mex.

Even a very large successor SG value need not make the current
board's SG value large. For example, the only successors of
$(1,100)$ are $(100)$ and $(1,99)$, with SG values $100$ and $0$.
Therefore
\[
 g(1,100)
 =\mex\{g(100),g(1,99)\}
 =\mex\{100,0\}
 =1.
\]

\subsection{Morisawa's bounds and periodicity}

\paragraph{The SG bound.}
Morisawa counts possible removals from a reference heap $h_i$
and from the other $m-1$ heaps. Each move supplies at most one
successor SG value, so the total number of moves bounds the mex.

\begin{center}
\begin{tabularx}{\linewidth}{@{}L{0.29\linewidth}Y@{}}
\toprule
Quantity and meaning & Expression and reason \\
\midrule
$h_i$\newline
Size of the reference heap
  & A removal must fit in the heap: $1\le d\le h_i$.
    There are therefore at most $h_i$ possible removals. \\
\addlinespace[9pt]

$\tau(h_i)$\newline
Number of positive divisors of $h_i$
  & Every removal from another heap must divide $h_i$.
    There are therefore at most $\tau(h_i)$ possible removal
    amounts from each other heap. \\
\addlinespace[9pt]

$(m-1)\tau(h_i)$\newline
Bound for all other heaps
  & There are $m-1$ other heaps, each permitting at most
    $\tau(h_i)$ removals. Multiplying gives their combined bound. \\
\bottomrule
\end{tabularx}
\end{center}

Adding the contributions and choosing the reference heap that
gives the smallest bound yields Morisawa's general SG bound
\cite[Lemma 5.2]{Morisawa2026}:
\begin{equation}
 \boxed{
 g(P)\le\min_i\bigl[h_i+(m-1)\tau(h_i)\bigr].
 }
 \label{eq:oldbound}
\end{equation}

For a simpler expression, use the smallest heap, $h_{\min}$.
Since $\tau(h_{\min})\le h_{\min}$,
\[
 g(P)\le h_{\min}+(m-1)h_{\min}=mh_{\min}.
\]

\paragraph{The period bound.}
The board $(A,n)$ consists of fixed heaps $A$ and one varying
heap of size $n$. An \emph{eventual period} $\omega_A$ is a
positive increase in the varying heap size that preserves the
SG value from some starting size $n$ onward:
\[
 \boxed{
 g(A,n+t)=g(A,n+\omega_A+t)
 \qquad\text{for every integer }t\ge0.
 }
\]
The smallest such positive increase is the \emph{least eventual period}.
We use $\omega_A$ for this least period when stating divisibility claims;
any positive multiple of it is also an eventual period.
Morisawa proves that an eventual period exists whenever at least
one fixed heap is nonempty
\cite[Theorem 2.5]{Morisawa2026}.

For the construction below, $A=(h_1,\ldots,h_{m-1})$ contains
at least two nonempty heaps. The number $m$ includes the varying
heap, and $h_{\min}$ is the smallest fixed heap.

Morisawa's period bound combines the number of possible blocks
of consecutive SG values with the comparison spacing $\omega'$.

A move on the varying heap reaches $(A,n-d)$, using an earlier
SG value supplied by the block. A move on a fixed heap produces
a smaller fixed board $Q$ and reaches $(Q,n)$. The comparison
spacing makes legal moves and these smaller-board SG values
repeat together.

\begin{center}
\begin{tabularx}{\linewidth}{@{}L{0.29\linewidth}Y@{}}
\toprule
Quantity and meaning & Expression and reason \\
\midrule
$0,\ldots,mh_{\min}$\newline
SG value range
  & The SG bound gives $g(A,n)\le mh_{\min}$.
    Including zero gives $mh_{\min}+1$ possibilities
    for each entry in a block. \\
\addlinespace[9pt]

$M_A$\newline
Increase in $n$ that repeats legal moves
  & Omit each fixed heap in turn, take the GCD of the remaining
    fixed heaps, then take their LCM:
    \[
      M_A=\operatorname{lcm}_{1\le i\le m-1}
           \gcd\{h_j:1\le j\le m-1,\ j\ne i\}.
    \]
    Every possible removal size divides $M_A$. Thus
    $d\mid n$ exactly when $d\mid n+M_A$, and increasing
    a sufficiently large $n$ by $M_A$ preserves legal moves. \\
\addlinespace[9pt]

$M_A+1$\newline
Number of consecutive SG values compared
  & Morisawa compares blocks
    \[
      \bigl(g(A,n),\ldots,g(A,n+M_A)\bigr).
    \]
    Every removal from the varying heap is at most $M_A$,
    so these entries cover all earlier values needed
    to continue the comparison. \\
\addlinespace[9pt]

$\omega_{Q_1},\ldots,\omega_{Q_t}$\newline
SG periods for all smaller fixed boards
  & The list $Q_1,\ldots,Q_t$ contains every nonempty fixed
    board with fewer total stones than $A$ and at most
    $m-1$ heaps. Morisawa first establishes an eventual
    period $\omega_{Q_i}$ for each sequence $g(Q_i,n)$.
    This covers the values reached by moves on fixed heaps. \\
\addlinespace[9pt]

$\omega'$\newline
Comparison spacing
  & The required common multiple is
    \[
      \omega'=\operatorname{lcm}
      \bigl(M_A,\omega_{Q_1},\ldots,\omega_{Q_t}\bigr).
    \]
    It repeats legal moves and all the smaller-board
    SG values together, once their repeating parts have begun. \\
\bottomrule
\end{tabularx}
\end{center}

There are at most $(mh_{\min}+1)^{M_A+1}$ different blocks.
Comparing one more block than this number, at intervals of
$\omega'$, guarantees two matching blocks. Both types of moves
then supply matching SG values, so the agreement continues
indefinitely. Their separation is an eventual period, giving
Morisawa's counting bound
\cite[Section 8]{Morisawa2026}:
\begin{equation}
 \boxed{
 \omega_A\le
 \underbrace{\omega'}_{\text{comparison spacing}}
 \underbrace{(mh_{\min}+1)^{M_A+1}}_{\text{number of possible blocks}}.
 }
 \label{eq:oldperiodbound}
\end{equation}

\paragraph{Our contribution.}
We obtain an SG bound that depends only on the smallest heap,
with no dependence on the number or sizes of the remaining heaps.

For periodicity, the smallest fixed heap bounds the SG value
range, the number of consecutive SG values compared, and
an increase in $n$ that repeats legal moves. The remaining
dependence on the other heaps is confined to the SG periods
of direct successor boards.

We then replace these unknown periods with a common-period construction.
The largest fixed heap keeps the same bounds valid throughout the
successor chains. Each recursive layer can add a bounded multiplier;
the total fixed stones bound how many layers are needed. This gives
an explicit period and a fixed finite set of possible prime divisors.
Their exponents may still grow with the number of layers.
The construction also gives a period bound for joined boards.

The title refers to removing the heap-count parameter $m$ from
Morisawa's bounds. We achieve this for the SG value. For the period,
we remove $m$ from the local terms and identify the remaining
dependence through recursive layers. Our conjecture proposes the
common period $2\operatorname{lcm}(1,\ldots,h_{\max})$, depending
only on the largest fixed heap.
Morisawa's two papers provide the starting point
\cite{Morisawa2025,Morisawa2026}; the odd-heap cases lead to our
bounds for all-even boards with arbitrarily many heaps.

\section{Bounding Grundy values}
\label{sec:values}

The SG value is bounded by the number of distinct successor
values, rather than the number of moves. Many moves reaching
the same value contribute only once.

We start with positions containing odd heaps. For all-even
positions, we then classify removals by their factors of $2$.
Some moves reach positions already bounded, others reach value
$0$, and the remaining moves are few enough to count. Adding
these contributions gives the SG bound.

\subsection{The starting cases: odd heaps}

An unchanged odd heap forbids every even removal. Therefore,
when at least two heaps are odd, every legal removal is odd
and changes the number of odd heaps by one.

\begin{lemma}[Odd heaps]
\label{lem:odd}
Among positions containing an odd heap, those with an even
number of odd heaps are losing. Those with an odd number are
winning. When at least two heaps are odd, the exact SG values are:
\begin{center}
\begin{tabularx}{0.9\linewidth}{@{}X c@{}}
\toprule
Number of odd heaps & SG value \\
\midrule
Even, and at least two & $0$ \\
Odd, and at least three & $1$ \\
\bottomrule
\end{tabularx}
\end{center}
\end{lemma}

\begin{proof}
The outcome statement is the odd-heap case of Morisawa's
criterion \cite[Theorem 2.2]{Morisawa2026}.
From an odd number of at least three odd heaps, every move
leaves a positive even number of odd heaps. Every successor
therefore has value $0$, so the current value is
$\mex\{0\}=1$.
\end{proof}

Exactly one odd heap requires a bound rather than this exact
formula. Moves on any other heap create a second odd heap and
reach value $0$. Only moves on the odd heap can contribute
additional values.

The symbol $H$ denotes a fixed upper limit on the size of one
positive heap. Keeping $H$ fixed lets the same bound apply after
that heap shrinks. Every legal removal then satisfies
\begin{equation}
 d\le H.
 \label{eq:smallremoval}
\end{equation}
Indeed, a removal from that small heap must fit inside it.
A removal from another heap must divide the unchanged small
heap, so it also cannot exceed its size.

\begin{lemma}[The starting bound]
\label{lem:base}
If a position contains an odd heap and a positive heap of
size at most $H$, then
\[
 g(P)\le H+1.
\]
\end{lemma}

\begin{proof}
If at least two heaps are odd, Lemma~\ref{lem:odd} gives
$g(P)\le1$.

If exactly one heap is odd and other heaps remain, all moves
on those other heaps reach value $0$. There are at most $H$
removal sizes available on the odd heap, by
\eqref{eq:smallremoval}. Thus there are at most $H+1$
distinct successor values, giving $g(P)\le H+1$.

If only one heap remains, its size is at most $H$, and its
SG value equals its size.
\end{proof}

\subsection{Classifying moves by 2-adic depth}

The \emph{2-adic depth} $v_2(a)$ of a positive integer $a$
is the number of factors of $2$ in it:
\[
 a=2^{v_2(a)}u,
 \qquad u\text{ odd}.
\]
Equivalently, this counts how many times $a$ can be halved
before becoming odd.

For a nonempty position, its minimum depth is $k(P)$, and
the number of heaps at that depth is $r(P)$:
\begin{equation}
 k(P)=\min_i v_2(h_i),
 \qquad
 r(P)=\#\{i:v_2(h_i)=k(P)\}.
 \label{eq:depth}
\end{equation}
Dividing every heap by $2^{k(P)}$ exposes exactly $r(P)$
odd heaps.

The smallest heap need not have minimum depth. For example,
in $(4,6)$ the smallest heap is $4$, but $6$ has the smaller
2-adic depth.

\begin{proposition}[Morisawa's outcome criterion]
\label{prop:halving}
A nonempty position is losing exactly when the number of heaps
at minimum depth is even:
\begin{equation}
 \boxed{
 g(P)=0
 \quad\Longleftrightarrow\quad
 r(P)\text{ is even}.
 }
 \label{eq:outcome}
\end{equation}
Consequently, dividing all heaps of an all-even position by
$2$ preserves its winning or losing outcome.
\end{proposition}

\begin{proof}
The criterion is Morisawa's Theorem~2.2
\cite{Morisawa2026}. Common halving decreases every heap's
depth by one and leaves the number attaining the minimum
unchanged.
\end{proof}

Common halving need not preserve the full SG value. For example,
\[
 g(2,4)=3,
 \qquad
 g(1,2)=1.
\]
To bound the full value, we account for the different removals.

Consider an all-even winning position. Its minimum depth
$k=k(P)$ is positive, and $r(P)$ is odd. A removal of size
$d$ has depth $j=v_2(d)$. Its contribution is determined by
how $j$ compares with $k$ and which heap is changed.

\begin{table}[H]
\centering
\begin{tabularx}{\linewidth}{@{}L{0.30\linewidth}Y@{}}
\toprule
Condition and meaning & Successor and reason \\
\midrule
$j<k$\newline
Removal below the minimum depth
  & The changed heap satisfies $v_2(h_i-d)=j$.
    After division by $2^j$, an even integer loses an odd
    integer, leaving an odd remainder. The successor has
    lower minimum depth, where we use the earlier bound. \\
\addlinespace[9pt]
$j=k$\newline
A minimum-depth heap remains unchanged
  & Division by $2^k$ makes the removal odd, so the count
    at depth $k$ changes by one. An unchanged heap keeps
    that count positive. It is now even, so the successor
    has SG value $0$. \\
\addlinespace[9pt]
$r(P)=1,\quad 2^k\mid d$\newline
Removal from the unique minimum-depth heap
  & The preceding two arguments do not bound this successor's
    value. We instead count the moves: each removal amount
    contributes at most one additional SG value. \\
\bottomrule
\end{tabularx}
\caption{The three types of moves from an all-even winning position.}
\label{tab:moves}
\end{table}

A removal with $j>k$ cannot divide an unchanged heap at depth
$k$. It must therefore change the unique heap at that depth,
so it belongs to the third row. The three rows cover every move.

\subsection{Counting the contributions using one small heap}

The fixed size limit $H$ controls both parts of the argument.
It bounds the removal sizes, and a positive heap of size at
most $H$ survives every move to which we apply the lower-depth
bound.

To see this, suppose the current position contains a heap
$a\le H$ and a removal has depth $j<k$. If another heap is
changed, $a$ remains. If $a$ is changed, deleting it would
require $d=a$, which is impossible because
\[
 v_2(d)=j<k\le v_2(a).
\]
A positive remainder smaller than $a$ therefore survives.

Now consider the third row of Table~\ref{tab:moves}.
These moves all act on one particular heap: the unique heap
of minimum depth. Their removal sizes are multiples of
$2^k$ and cannot exceed $H$. Hence
\begin{equation}
 \text{number of these moves}
 \le
 \left\lfloor\frac{H}{2^k}\right\rfloor.
 \label{eq:exceptioncount}
\end{equation}
Each size gives at most one move because the target heap is
fixed. This is why the count has no factor for the number
of heaps.

These moves may delete the small heap or reach a board with
a large SG value. We do not calculate their successor values:
each move contributes at most one value.

Write $U_j(H)$ for the bound we construct at minimum depth
$j$, for positions containing a positive heap of size at
most $H$. The odd-heap result supplies the starting value
\[
 U_0(H)=H+1.
\]

At depth $k\ge1$, we add two counts: the possible values already
covered at lower depths, including $0$, and the additional values
from removals on the unique minimum-depth heap.

\begin{table}[H]
\centering
\begin{tabularx}{\linewidth}{@{}L{0.30\linewidth}Y@{}}
\toprule
Quantity and meaning & Expression and reason \\
\midrule
$U_{k-1}(H)+1$\newline
Number of values in the lower-depth range
  & The lower-depth successors and losing successors have
    values in $0,\ldots,U_{k-1}(H)$. Including zero gives
    at most $U_{k-1}(H)+1$ distinct values. \\
\addlinespace[9pt]
$\lfloor H/2^k\rfloor$\newline
Number of multiples of $2^k$ up to $H$
  & These removals are multiples of $2^k$ no larger than $H$.
    There are at most $\lfloor H/2^k\rfloor$ such amounts.
    Each acts on the same heap and contributes at most one
    successor value. \\
\bottomrule
\end{tabularx}
\caption{The two contributions to the bound at depth $k$.}
\label{tab:recursionterms}
\end{table}

The mex is at most the number of distinct successor values.
Adding the two contributions gives
\begin{equation}
 U_0(H)=H+1,
 \qquad
 \boxed{
 U_k(H)=
 \underbrace{U_{k-1}(H)+1}_{\text{lower-depth range, including }0}
 +\underbrace{\left\lfloor\frac{H}{2^k}\right\rfloor}_{
 \text{additional values}}.
 }
 \label{eq:recursion}
\end{equation}

The recursion follows decreasing 2-adic depth, not a sequence
of moves in actual play. We apply it only to the first row
of Table~\ref{tab:moves}; the other rows are accounted for
without further recursion.

\begin{theorem}[A bound independent of heap count]
\label{thm:bound}
Every nonempty position containing a heap of size at most $H$
satisfies
\[
 g(P)\le U_{k(P)}(H).
\]
If the number $r(P)$ of heaps at minimum depth is odd and at
least three, then
\begin{equation}
 g(P)\le
 U_{k(P)}(H)
 -
 \left\lfloor\frac{H}{2^{k(P)}}\right\rfloor.
 \label{eq:refined}
\end{equation}
\end{theorem}

\begin{proof}
A sole heap has value equal to its size, at most $H$, so it
satisfies the bound at every depth.

For positions with more heaps, induct on the minimum depth.
At depth zero, Lemma~\ref{lem:base} gives
$g(P)\le H+1=U_0(H)$.

At depth $k\ge1$, a losing position has value $0$.
For a winning position, every lower-depth successor retains
a positive heap of size at most $H$. By induction and the
fact that $U_j(H)$ increases with $j$, its value is at most
$U_{k-1}(H)$. The other moves contribute exactly the allowances
in Table~\ref{tab:recursionterms}. Their sum gives
$g(P)\le U_k(H)$.

If $r(P)\ge3$ is odd, there is no unique heap of minimum depth.
The second allowance in that table is therefore absent, giving
the refinement for $k\ge1$. At depth zero, Lemma~\ref{lem:odd}
gives
\[
 g(P)=1=U_0(H)-H,
\]
which proves the refinement there as well.
\end{proof}

\subsection{The bound from the smallest heap}

Expanding the recurrence adds the allowances from each depth:
\begin{equation}
 U_k(H)=
 H+\left\lfloor\frac H2\right\rfloor
 +\cdots+
 \left\lfloor\frac H{2^k}\right\rfloor
 +k+1.
 \label{eq:B}
\end{equation}
The first terms count the possible additional values at the
successive depths. The term $k+1$ comes from the $+1$ at the
base and at each recursive step.

Now choose $H=h_{\min}$, the smallest nonempty heap of the
original position. Its depth is at least the minimum depth:
\[
 k=k(P)\le v_2(h_{\min}).
\]
Thus $2^k$ divides $h_{\min}$, and every floor in
\eqref{eq:B} is an exact division. The contributions combine as
\[
 U_k(h_{\min})
 =
 \underbrace{
 h_{\min}+\frac{h_{\min}}2+\cdots+
 \frac{h_{\min}}{2^k}
 }_{\text{allowances from successive depths}}
 +
 \underbrace{k+1}_{\text{the base and recursive }+1\text{ terms}}.
\]

\begin{corollary}[The bound from the smallest heap]
\label{cor:closed}
For a nonempty position with smallest heap $h_{\min}$ and
minimum 2-adic depth $k$,
\begin{equation}
 \boxed{
 \begin{aligned}
 g(P)
 &\le U(h_{\min},k)
 :=2h_{\min}-\frac{h_{\min}}{2^k}+k+1 \\
 &\le C(h_{\min})
 :=2h_{\min}+\lfloor\log_2h_{\min}\rfloor.
 \end{aligned}
 }
 \label{eq:closed}
\end{equation}
If $r(P)\ge3$ is odd, the first bound improves by
$h_{\min}/2^k$.
\end{corollary}

\begin{proof}
The geometric sum gives the first expression.
The second inequality follows from
\[
 \frac{h_{\min}}{2^k}\ge1,
 \qquad
 k\le\lfloor\log_2h_{\min}\rfloor.
\]
The refinement follows from \eqref{eq:refined}.
\end{proof}

For example, when the smallest heap has size $4$, the bound
is built as follows:

\begin{table}[H]
\centering
\small
\begin{tabular}{@{}c c c@{}}
\toprule
Minimum 2-adic depth & Calculation & SG bound \\
\midrule
$0$ & $4+1$ & $5$ \\
$1$ & $5+4/2+1$ & $8$ \\
$2$ & $8+4/4+1$ & $10$ \\
\bottomrule
\end{tabular}
\caption{Building the bound when $h_{\min}=4$.}
\label{tab:examplebound}
\end{table}

At depth $2$, if the number of heaps at minimum depth is odd
and at least three, the allowance $4/4=1$ is absent and the
bound improves from $10$ to $9$.

The ceiling $C(h_{\min})$ depends only on the smallest heap.
Neither the number nor the sizes of the other heaps appear.

\begin{proposition}[A bound for every descendant]
\label{prop:descendantbound}
Let $P$ be a nonempty position and define
\[
 h_{\max}(P):=\max_i h_i.
\]
Every nonempty position $R$ reached from $P$ by any number of
legal moves satisfies
\begin{equation}
 \boxed{
 g(R)\le C\bigl(h_{\max}(P)\bigr)
 =
 2h_{\max}(P)
 +\left\lfloor\log_2 h_{\max}(P)\right\rfloor.
 }
 \label{eq:descendantbound}
\end{equation}
The empty position has SG value $0$.
\end{proposition}

\begin{proof}
A legal move never increases a heap. Therefore every heap in a
nonempty descendant $R$ has size at most $h_{\max}(P)$, and in
particular
\[
 h_{\min}(R)\le h_{\max}(P).
\]
Corollary~\ref{cor:closed} gives
\[
 g(R)\le C\bigl(h_{\min}(R)\bigr)
 \le C\bigl(h_{\max}(P)\bigr),
\]
because $C$ is nondecreasing. The empty position has value $0$.
\end{proof}

\section{Bounding periods}
\label{sec:periods}

The fixed board is $A=(h_1,\ldots,h_{m-1})$, and $n$ is the
varying heap. The smallest fixed heap is $h_{\min}$; $m$ includes
the varying heap. We use the eventual period $\omega_A$ defined
in Section~\ref{sec:intro}: from some starting size onward,
increasing $n$ by $\omega_A$ preserves every SG value.

Our bound has the same structure as Morisawa's: the number of
possible blocks of consecutive SG values multiplied by the
comparison spacing, now denoted by $W$. We bound each component
using the smallest fixed heap, apart from the SG periods of
direct successor boards.

\subsection{The ingredients of the period bound}

As before, a move changes either the varying heap or a fixed heap.
Write $A\xrightarrow{d}Q$ when removing $d$ stones from one
heap of $A$ produces $Q$ and $d$ divides every other heap of $A$.
On the full board $(A,n)$, the removal must also divide $n$.

\begin{center}
\begin{tabularx}{\linewidth}{@{}L{0.30\linewidth}Y@{}}
\toprule
Successor value and meaning & Expression and reason \\
\midrule
$g(A,n-d)$\newline
Value after changing the varying heap
  & The removal must divide every fixed heap. Thus $d\mid D$,
    where $D=\gcd(h_1,\ldots,h_{m-1})$. This successor uses
    an earlier SG value in the same sequence. \\
\addlinespace[9pt]
$g(Q,n)$\newline
Value after changing a fixed heap
  & The move must satisfy $A\xrightarrow{d}Q$ and $d\mid n$.
    Its SG value comes from the smaller fixed board $Q$
    with the varying heap attached. \\
\bottomrule
\end{tabularx}
\end{center}

Combining the two types of successor values gives, for $n\ge D$,
\begin{equation}
 g(A,n)=\mex\left(
 \{g(A,n-d):d\mid D\}\ \cup\
 \{g(Q,n):A\xrightarrow{d}Q,\ d\mid n\}\right).
 \label{eq:periodrec}
\end{equation}

The largest removal from $n$ is $D$, so a block of $D$
consecutive SG values supplies every earlier value needed for
the next calculation. Each entry has at most $C(h_{\min})+1$
possibilities. The comparison spacing $W$ makes legal moves
and the SG values of direct successor boards repeat together.

\begin{center}
\begin{tabularx}{\linewidth}{@{}L{0.30\linewidth}Y@{}}
\toprule
Quantity and meaning & Expression and reason \\
\midrule
$0,\ldots,C(h_{\min})$\newline
SG value range
  & Our SG bound gives
    $C(h_{\min})=2h_{\min}+\lfloor\log_2h_{\min}\rfloor$.
    Including zero gives $C(h_{\min})+1$ possibilities
    for each entry in a block. \\
\addlinespace[9pt]
$J(h_{\min}),\ J_A$\newline
Increases in $n$ that repeat legal moves
  & Every fixed-heap removal is at most $h_{\min}$, so
    $J(h_{\min})=\operatorname{lcm}(1,\ldots,h_{\min})$
    is a multiple of every such removal. With at least two
    fixed heaps, every removal also divides Morisawa's $M_A$.
    We therefore use the greatest common divisor
    \begin{equation}
      J_A=\gcd\bigl(M_A,J(h_{\min})\bigr),
      \label{eq:legality}
    \end{equation}
    which is a common multiple of every legal fixed-heap removal
    amount. With one fixed heap, use $J_A=J(h_{\min})$. \\
\addlinespace[9pt]
$D$\newline
Number of consecutive SG values compared
  & The largest varying-heap removal is
    $D=\gcd(h_1,\ldots,h_{m-1})\le h_{\min}$.
    The block
    \[
      \bigl(g(A,n-D+1),\ldots,g(A,n)\bigr)
    \]
    supplies all earlier values needed to compute $g(A,n+1)$. \\
\bottomrule
\end{tabularx}
\end{center}

These three quantities are bounded using only $h_{\min}$. The SG periods
of direct successor boards complete the comparison spacing:

\begin{center}
\begin{tabularx}{\linewidth}{@{}L{0.30\linewidth}Y@{}}
\toprule
Quantity and meaning & Expression and reason \\
\midrule
$\omega_{Q_1},\ldots,\omega_{Q_t}$\newline
SG periods for direct successor boards
  & The list contains only the distinct nonempty boards
    $Q$ reached by $A\xrightarrow{d}Q$.
    Each $\omega_Q$ is a chosen eventual period of $g(Q,n)$.
    No other smaller fixed boards supply values in the recurrence. \\
\addlinespace[9pt]
$W$\newline
Comparison spacing
  & The required common multiple is
    \[
      W=\operatorname{lcm}(J_A,\omega_{Q_1},\ldots,\omega_{Q_t}).
    \]
    It repeats legal moves and all direct-successor SG values
    together, once their repeating parts have begun. \\
\bottomrule
\end{tabularx}
\end{center}

If a move deletes the last fixed heap, its successor is the sole
heap $n$, with SG value $n$. For $n>C(h_{\min})$, this value is
too large to affect the parent's mex. Thus no period is needed
for an empty fixed board. If there are no nonempty direct
successors, take $W=J_A$.

\begin{lemma}[Repeating legal moves]
\label{lem:legality}
For $n\ge h_{\min}$, increasing $n$ by $J_A$ preserves every
legal choice of heap and removal amount.
\end{lemma}
\begin{proof}
A removal from a smallest fixed heap must fit in that heap.
A removal from any other fixed heap must divide the unchanged
smallest heap. In both cases $d\le h_{\min}$, so
$d\mid J(h_{\min})$.

With at least two fixed heaps, a removal from $h_i$ also divides
the GCD of the other fixed heaps. It therefore divides $M_A$
and hence $J_A$. Thus $d\mid n$ exactly when $d\mid n+J_A$.
With one fixed heap, $d\mid J_A$ holds by definition.
For the varying heap, the permitted removals are the divisors
of $D\le h_{\min}$, which already fit.
\end{proof}

There are at most $[C(h_{\min})+1]^D$ different blocks.
Comparing one more block than this number, at intervals of
$W$, guarantees two matching blocks. Their separation is an
eventual period, as the following proof shows.

\begin{theorem}[An eventual-period bound]
\label{thm:period}
There is an eventual period satisfying
\begin{equation}
 \boxed{\omega_A\le
 \underbrace{W}_{\text{comparison spacing}}
 \underbrace{[C(h_{\min})+1]^D}_{\text{number of possible blocks}}.}
 \label{eq:periodfine}
\end{equation}
In particular,
\begin{equation}
 \boxed{\omega_A\le
 \underbrace{\operatorname{lcm}(\omega_{Q_1},\ldots,\omega_{Q_t})}_{
 \text{SG periods of direct successors}}
 \underbrace{J(h_{\min})[C(h_{\min})+1]^{h_{\min}}}_{
 \text{depends only on }h_{\min}}.}
 \label{eq:periodseparated}
\end{equation}
The LCM of an empty list is taken to be $1$.
\end{theorem}
\begin{proof}
Write $N=[C(h_{\min})+1]^D$. Choose $n$ beyond the initial
segments of all successor periods, with $n>\max(D,C(h_{\min}))$.
Inspect the $N+1$ ordered blocks of $D$ consecutive SG values
ending at $n,n+W,\ldots,n+NW$. Two blocks agree; their endpoints
differ by $qW$ for some $1\le q\le N$.

At the next heap size, the fixed-heap moves supply the same values
because their legality and successor values repeat after $qW$.
The varying-heap moves refer to matching entries in the two blocks.
Thus the mex gives the same next value. Repeating the argument
proves that $qW$ is an eventual period. Finally,
$W\le J(h_{\min})\operatorname{lcm}(\omega_{Q_1},\ldots,\omega_{Q_t})$
and $D\le h_{\min}$ give the separated bound.
\end{proof}

\subsection{Comparison with Morisawa}

Morisawa's counting argument uses a larger SG range and longer
blocks. His comparison spacing $\omega'$ also includes periods
of all nonempty fixed boards with fewer total stones, whereas $W$
uses only direct successors. Assign the same chosen period to each
board in both constructions. Then $W\mid\omega'$.

\begin{table}[htbp]
\centering\small
\begin{tabularx}{\linewidth}{@{}L{0.28\linewidth}YY@{}}
\toprule
Quantity & Morisawa's proof & Present bound\\
\midrule
SG value range & $0,\ldots,mh_{\min}$ & $0,\ldots,C(h_{\min})$\\
\addlinespace[7pt]
Number of consecutive SG values compared & $M_A+1$ & $D\le h_{\min}$\\
\addlinespace[7pt]
Number of possible blocks & $(mh_{\min}+1)^{M_A+1}$ &
$[C(h_{\min})+1]^D$\\
\addlinespace[7pt]
Increase in $n$ that repeats legal moves & $M_A$ & $J_A\mid J(h_{\min})$\\
\addlinespace[7pt]
SG periods included & All nonempty fixed boards with
fewer total stones, using at most $m-1$ heaps & Only the nonempty
direct successors $Q_1,\ldots,Q_t$\\
\addlinespace[7pt]
Comparison spacing & $\omega'$ & $W$\\
\addlinespace[7pt]
Period upper bound & $\omega'(mh_{\min}+1)^{M_A+1}$ &
$W[C(h_{\min})+1]^D$\\
\bottomrule
\end{tabularx}
\caption{Comparison with the counting argument in Morisawa's proof
of Theorem 2.5 \cite[Section 8]{Morisawa2026}. Here $m\ge3$ includes
the varying heap.}
\label{tab:periods}
\end{table}

Since $C(h_{\min})<3h_{\min}\le mh_{\min}$ and $D\le M_A$,
both the SG value range and the number of consecutive SG values compared are reduced. Equation
\eqref{eq:periodseparated} shows exactly what remains: in that
uniform bound, only the LCM of successor-board periods is not
controlled by $h_{\min}$.

For a concrete comparison, take $A=(2,3)$. Its direct successors
are $(1,3)$, $(2,2)$, and $(1,2)$, with periods $2,4,2$.
Thus $h_{\min}=2$, $D=1$, $C(h_{\min})=5$, and
\[
 J_A=2,\qquad W=\operatorname{lcm}(2,2,4,2)=4.
\]
Morisawa's larger list also includes the single fixed heap $(4)$,
since $4<2+3$. Its sequence has period $24$; the other smaller
boards have periods $2$ or $4$. Using these least periods gives
$M_A=6$ and $\omega'=24$. The comparison is therefore
\[
\begin{array}{rcl}
 \text{Morisawa's counting bound}&=&24\cdot7^7=19{,}765{,}032,\\[3pt]
 \text{Present bound}&=&4\cdot6^1=24,\\[3pt]
 \text{Actual least period}&=&4.
\end{array}
\]

\subsection{Keeping the same bounds throughout the successor chains}

The smallest fixed heap may disappear. We therefore need a size
limit that stays valid after any number of moves on the fixed board.
The largest original fixed heap supplies that limit: no move can
increase a heap.

Throughout this construction, $h_{\max}$ is the largest fixed heap
in the original board $A$. It excludes the varying heap $n$ and
stays fixed throughout the induction. Every nonempty fixed descendant $Q$ still contains a heap no
larger than $h_{\max}$, even when $n$ is arbitrarily large.
Thus $h_{\min}(Q,n)\le h_{\max}$.
Corollary~\ref{cor:closed}, together with the monotonicity of $C$,
therefore gives
\begin{equation}
 \boxed{g(Q,n)\le C(h_{\max})
 =2h_{\max}+\lfloor\log_2h_{\max}\rfloor.}
 \label{eq:descendantvalues}
\end{equation}
The current board can use its smaller bound $C(h_{\min})$.
The bound $C(h_{\max})$ keeps one size limit valid for every
nonempty fixed descendant.

The same replacement works for all three local ingredients of
the period argument:

\begin{table}[H]
\centering
\begin{tabularx}{\linewidth}{@{}L{0.30\linewidth}YY@{}}
\toprule
Quantity & Current fixed board $A$ & Nonempty fixed descendant $Q$ \\
\midrule
SG value ceiling
 & $C(h_{\min})$
 & $C(h_{\max})$ \\
\addlinespace[9pt]
Number of previous SG values needed
 & $D_A=\gcd(A)\le h_{\min}$
 & $D_Q=\gcd(Q)\le h_{\max}$ \\
\addlinespace[9pt]
Increase in $n$ that repeats legal moves
 & $J(h_{\min})$ suffices
 & $J(h_{\max})$ suffices \\
\bottomrule
\end{tabularx}
\caption{One set of bounds remains valid even after the smallest
fixed heap disappears. The varying heap $n$ is not included in
$h_{\max}$.}
\label{tab:uniformperiodterms}
\end{table}

Each descendant has at least one positive fixed heap at most
$h_{\max}$. This bounds its SG value, its largest varying-heap
removal, and every fixed-heap removal. Once $n\ge h_{\max}$,
all varying-heap removals fit, and adding $J(h_{\max})$ preserves
the divisibility tests on $n$.

If a move removes the last fixed heap, the successor has SG value
$n$. For $n>C(h_{\max})$, that value exceeds the parent's ceiling
and cannot affect its mex. Thus only nonempty fixed successors
need periods.

\subsection{From a block bound to one common period}

At each board, the earlier argument compares blocks of consecutive
SG values. The table above bounds both the length of a block and
the number of choices for each entry. It therefore bounds the
number of different blocks using $h_{\max}$ alone.

Suppose $T$ repeats all direct-successor SG sequences and legal
moves. Two matching blocks, compared at intervals of $T$, give a
parent period $qT$. The block count bounds $q$. To cover every
possible $q$, we take the LCM of all integers up to that bound.

\begin{table}[H]
\centering
\begin{tabularx}{\linewidth}{@{}L{0.30\linewidth}Y@{}}
\toprule
Quantity and meaning & Expression and reason \\
\midrule
$N_{\max}$\newline
Bound on the number of possible blocks
  & Each entry has at most $C(h_{\max})+1$ possible SG values,
    and at most $h_{\max}$ entries are needed. Thus
    \[
      N_{\max}=[C(h_{\max})+1]^{h_{\max}}.
    \]
    Comparing $N_{\max}+1$ blocks gives two matching blocks
    separated by $qT$, where $1\le q\le N_{\max}$. \\
\addlinespace[9pt]
$K(h_{\max})$\newline
Common multiple of all possible multipliers
  & The multiplier $q$ need not divide $N_{\max}$. Instead,
    use a number divisible by every possible $q$:
    \[
      K(h_{\max})=\operatorname{lcm}(1,\ldots,N_{\max})
      =J(N_{\max}).
    \]
    Every possible period $qT$ then divides $T K(h_{\max})$. \\
\bottomrule
\end{tabularx}
\caption{A bound on the block count gives a common multiple of
all periods that one recursive layer may produce.}
\label{tab:commonmultipliers}
\end{table}

The distinction between a bound and a common multiple matters.
For example, $3\le4$ does not mean that $3$ divides $4$.
Although a parent period is at most $T N_{\max}$, that number
need not itself be a period. The number $T K(h_{\max})$ is a
period because it is a multiple of every possible $qT$.

\begin{lemma}[One common period for the next layer]
\label{lem:commonlayer}
Suppose every heap of a nonempty fixed board is at most
$h_{\max}$, and $T$ is a positive multiple of $J(h_{\max})$.
If $T$ is an eventual period of every nonempty direct-successor
sequence, then $T K(h_{\max})$ is an eventual period of the
parent sequence.
\end{lemma}
\begin{proof}
Use $T$ as the comparison spacing in
Theorem~\ref{thm:period}. The actual block count is at most
$N_{\max}$, so the proof gives a period $qT$ with
$1\le q\le N_{\max}$. Since $q\mid K(h_{\max})$,
$T K(h_{\max})$ is a multiple of this period and is a period
too. This argument applies after all successor sequences have
reached their repeating parts.
\end{proof}

Within one layer, the number of successor boards adds no further
factor. If they all have a common period $T$, their values all
repeat after $T$, whether there are three successors or three
thousand. The next factor comes from continuing the parent's
own recurrence, rather than from counting these successors.

\subsection{Why the recursive bound contains a power}

LCM combines the successor periods at one layer. Moving up to a
parent can introduce a new multiplier \emph{outside} that LCM.
This is why a bound on one layer does not yet bound every layer
by the same number.

First consider a chain in which each layer uses the same multiplier
$N$. This illustrates the accumulation; the actual multipliers
may be different. Write $J$ for the increase that repeats legality.
If the starting period is $w_0=JN$, then
\[
 \begin{aligned}
 w_1&=\operatorname{lcm}(J,JN)\,N=JN^2,\\
 w_2&=\operatorname{lcm}(J,JN^2)\,N=JN^3.
 \end{aligned}
\]
The LCM keeps the previous power of $N$, because $JN^2$ already
contains $J$. The new multiplier then adds another factor of $N$.

For example, if $J=6$ and $N=4$, the calculation is:

\begin{table}[H]
\centering
\begin{tabular}{@{}l c r@{}}
\toprule
Layer & Calculation & Illustrative value \\
\midrule
Starting layer & $6\cdot4$ & $24$ \\
\addlinespace[6pt]
Next layer & $\operatorname{lcm}(6,24)\cdot4$ & $96$ \\
\addlinespace[6pt]
Following layer & $\operatorname{lcm}(6,96)\cdot4$ & $384$ \\
\bottomrule
\end{tabular}
\caption{Illustrative multiplier accumulation. These numbers
explain the recurrence; they are not computed Divisor Nim periods.}
\label{tab:layerillustration}
\end{table}

Repeated entries in $\operatorname{lcm}(J,N,N,N)$ add nothing.
In the constant-multiplier illustration above, each new layer
instead multiplies by another factor of $N$, giving
$JN,JN^2,JN^3$. These equalities are exact for that illustration,
and every term has the same prime divisors as $JN$.

The actual block argument does not establish that every multiplier
equals $N_{\max}$. If $T$ repeats legal moves and all direct-successor
SG sequences, the argument produces a parent period $qT$, where
\[
 1\le q\le N_{\max},
 \qquad qT\le T N_{\max}.
\]
This inequality bounds the period's size. It does not imply that
$qT$ divides $T N_{\max}$, or that its prime divisors divide
$T N_{\max}$.

To obtain a guaranteed common period, we use
\[
 K(h_{\max})=\operatorname{lcm}(1,\ldots,N_{\max}).
\]
Since every possible $q$ divides $K(h_{\max})$, the period $qT$
divides $T K(h_{\max})$. Thus $T K(h_{\max})$ is an eventual
period regardless of which multiplier occurs. This divisibility
is what supports the common-period induction and the subsequent
prime-factor conclusion.

We now build that guaranteed period layer by layer. The fixed
board $(1)$ starts the induction. Each move on a fixed heap
removes at least one stone, so the total fixed stones minus one
bounds the number of layers before only one fixed stone remains.

\begin{table}[H]
\centering
\begin{tabularx}{\linewidth}{@{}L{0.30\linewidth}Y@{}}
\toprule
Quantity and meaning & Expression and reason \\
\midrule
$T_0(h_{\max})$\newline
Starting common period
  & The sequence $g(1,n)$ has period $2$. Combine it with
    the increase that repeats legal moves:
    \[
      T_0(h_{\max})=\operatorname{lcm}(2,J(h_{\max})).
    \]
    This is a period for $(1,n)$ and repeats legality
    for every fixed board under the same size limit. \\
\addlinespace[9pt]
$s(A)$\newline
Number of recursive layers allowed
  & For a nonempty fixed board,
    \[
      s(A)=\sum_{h\in A}h-1.
    \]
    Repeated heaps are counted separately. A nonempty
    fixed successor $Q$ satisfies $s(Q)\le s(A)-1$.
    The case $s=0$ is exactly the board $(1)$.
    This counts fixed stones, excluding $n$; it is not
    the number of common halvings. \\
\bottomrule
\end{tabularx}
\caption{The base of the induction and a bound on its length.}
\label{tab:periodlayers}
\end{table}

For a fixed value of $h_{\max}$, the common periods build as follows:

\begin{table}[H]
\centering
\begin{tabularx}{\linewidth}{@{}c Y L{0.30\linewidth}@{}}
\toprule
Layer bound & Why it covers the next board & Guaranteed common period \\
\midrule
$0$ & The fixed board is $(1)$.
 & $T_0(h_{\max})$ \\
\addlinespace[8pt]
$1$ & All nonempty successors are covered at layer $0$;
one more layer contributes a factor $K(h_{\max})$.
 & $T_0(h_{\max})K(h_{\max})$ \\
\addlinespace[8pt]
$2$ & All nonempty successors are covered by the preceding
common period; apply the same one-layer argument.
 & $T_0(h_{\max})K(h_{\max})^2$ \\
\addlinespace[8pt]
$s$ & Every successor has a smaller layer bound, so
$s$ factors cover the whole induction.
 & $T_0(h_{\max})K(h_{\max})^s$ \\
\bottomrule
\end{tabularx}
\caption{The exponent counts repeated applications of the
one-layer period argument, not the number of successors at a layer.}
\label{tab:periodpowers}
\end{table}

Combining the starting period with one common multiplier per
allowed layer gives the explicit period.

\begin{proposition}[An explicit period without successor periods]
\label{prop:explicitperiod}
Fix a positive heap-size limit $h_{\max}$. For every nonempty
fixed board $A$ whose heaps are at most this limit,
\begin{equation}
 \boxed{
 \underbrace{T_0(h_{\max})}_{\text{starting common period}}
 \quad\cdot\quad
 \underbrace{K(h_{\max})^{s(A)}}_{\text{one factor per allowed layer}}
 }
 \label{eq:explicitperiod}
\end{equation}
is an eventual period of $g(A,n)$. Its least eventual period satisfies
\begin{equation}
 \boxed{\omega_A\mid T_0(h_{\max})K(h_{\max})^{s(A)}.}
 \label{eq:explicitdivisibility}
\end{equation}
\end{proposition}
\begin{proof}
Induct on $s=s(A)$, keeping the original size limit $h_{\max}$
fixed. If $s=0$, then $A=(1)$ and $g(1,n)$ is $1$ for even $n$
and $0$ for odd $n$. Thus $T_0(h_{\max})$ is a period.

For $s\ge1$, each nonempty direct successor $Q$ has
$s(Q)\le s-1$. By induction it has period
$T_0(h_{\max})K(h_{\max})^{s(Q)}$. All these periods divide
\[
 T=T_0(h_{\max})K(h_{\max})^{s-1}.
\]
Therefore $T$ repeats every nonempty successor sequence and
legal move. The one-layer lemma gives the parent period
\[
 T K(h_{\max})=T_0(h_{\max})K(h_{\max})^s.
\]
Finally, the least eventual period divides every eventual period:
on the repeating tail, any shift preserving the sequence must be
a multiple of its shortest cycle.
\end{proof}

For the original board, choose its actual largest fixed heap as
$h_{\max}$. The formula now uses only that heap and the total
fixed stones. It requires no separate calculation of successor
periods. Equation~\eqref{eq:periodfine} remains sharper when
those periods are known.

\subsection{What this proves about prime factors}

The number of layers can increase the \emph{powers} of prime
factors in this guaranteed period, but it cannot add primes
outside a fixed collection determined by $h_{\max}$.

Every prime dividing $K(h_{\max})=\operatorname{lcm}(1,\ldots,N_{\max})$
is at most $N_{\max}$. Also $N_{\max}\ge\max(2,h_{\max})$,
so $T_0(h_{\max})$ divides $K(h_{\max})$ and introduces no
other primes. Equation~\eqref{eq:explicitdivisibility} therefore gives:

\begin{corollary}[A fixed set of possible prime divisors]
\label{cor:primepool}
For fixed $h_{\max}$, every prime divisor of the least eventual
period of any board with fixed heaps at most $h_{\max}$ is at most
$N_{\max}=[C(h_{\max})+1]^{h_{\max}}$, independently of heap count.
\end{corollary}

Adding heaps no larger than $h_{\max}$ keeps this same collection
of possible primes. It can change which primes in the collection
actually occur. The proof does not show that a prime disappearing
from one period can never return, or that its exponent cannot grow.
For example, powers $2,4,8,16$ all have the same prime divisor
but grow without bound if the exponents keep increasing.

Long successor chains also need not produce long periods.
Any number of fixed heaps of size $1$ gives period $2$, although
$s(A)$ grows with their number. The exponent in our bound allows
growth at every layer; it does not say the actual period grows.

One possible improvement is to retain only the successor values
that affect the mex. Values above the parent's SG ceiling can be
ignored. A successor sequence can also be omitted if each relevant
value it supplies is already supplied by a retained legal move at
that same $n$. This condition must hold throughout the repeating
tail. Bounding how many layers remain after these reductions is
one route toward a period independent of heap count.

\subsection{A conjecture and computed periods}

The explicit bound leaves prime exponents dependent on the number
of layers. The computations suggest a much smaller common period:
twice the LCM of all positive integers up to the largest fixed heap.

\noindent\textbf{Conjecture.}
For every nonempty fixed board $A$, its least eventual SG period satisfies
\begin{equation}
 \boxed{\omega_A\mid2J(h_{\max}(A))
 =2\operatorname{lcm}(1,\ldots,h_{\max}(A)).}
 \label{eq:maxconjecture}
\end{equation}
Equivalently, $2J(h_{\max}(A))$ is an eventual period of $g(A,n)$.
The largest fixed heap excludes $n$. The starting size at which
the repetition begins may still depend on the whole fixed board.

This is stronger than asking for some bound $F(h_{\max})$:
it gives a particular common period. It also specifies the
possible prime powers. For an odd prime $p$, the LCM contains
the largest power $p^a\le h_{\max}$. Multiplication by $2$
allows one extra factor of $2$.

\begin{table}[H]
\centering
\begin{tabularx}{\linewidth}{@{}L{0.25\linewidth}YY@{}}
\toprule
Quantity & Proved here & Proposed in the conjecture \\
\midrule
Common eventual period
 & $T_0(h_{\max})K(h_{\max})^{s(A)}$
 & $2J(h_{\max})$ \\
\addlinespace[9pt]
Possible prime divisors
 & At most $N_{\max}$, independent of heap count
 & At most $h_{\max}$, apart from $2$ when $h_{\max}=1$ \\
\addlinespace[9pt]
Powers of those primes
 & The bound can grow with $s(A)$
 & Bounded using $h_{\max}$ alone \\
\bottomrule
\end{tabularx}
\caption{The conjecture would remove the remaining recursive-layer
dependence, not merely restrict the set of possible primes.}
\label{tab:periodconjecturecomparison}
\end{table}

As a baseline, the least eventual periods with one fixed heap are
listed below. Morisawa records the values for $h=1,\ldots,6$ from
earlier work and proves the value for $h=7$
\cite[Theorem 3 and Corollary 1]{Morisawa2025}. Our recurrence
calculation reproduces these entries and gives the displayed values
through $h=15$. These are least periods, not upper bounds.

\begin{table}[H]
\centering
\begin{tabular}{@{}c r@{}}
\toprule
Fixed heap $h$ & Least eventual period of $g(h,n)$\\
\midrule
$1$ & $2$\\
$2$ & $4$\\
$3$ & $2$\\
$4$--$6$ & $24$\\
$7$ & $120$\\
$8$--$10$ & $1{,}680$\\
$11$ & $5{,}040$\\
$12$--$14$ & $55{,}440$\\
$15$ & $720{,}720$\\
\bottomrule
\end{tabular}
\caption{Least eventual periods with one fixed heap. Ranges group
consecutive heap sizes having the same period.}
\label{tab:singleperiods}
\end{table}

Repetition was certified by checking a common shift that preserves
all divisibility tests and matching blocks for every nonempty fixed
descendant. Induction on fixed stones and then on $n$ extends the
matching blocks indefinitely. For $h\le15$, a common shift of
$720{,}720$ passes these checks from $n=10{,}000$. Testing the divisors
of that shift within a full cycle gives each least period above.

We also checked all fixed boards in the following finite ranges.
Heap order does not matter, but repeated sizes are included.

\begin{table}[H]
\centering
\begin{tabular}{@{}c c r@{}}
\toprule
Allowed fixed heap sizes & Maximum fixed heap count & Boards checked \\
\midrule
$1,\ldots,8$ & $12$ & $125{,}969$ \\
$1,\ldots,12$ & $6$ & $18{,}563$ \\
$1,\ldots,14$ & $4$ & $3{,}059$ \\
$1,\ldots,16$ & $3$ & $968$ \\
\bottomrule
\end{tabular}
\caption{Certified period checks supporting
\eqref{eq:maxconjecture}. The ranges overlap.}
\label{tab:periodcoverage}
\end{table}

After removing overlaps, these tests contain $143{,}059$ distinct
nonempty fixed boards. Every computed least period divides
$2J(h_{\max}(A))$, using that board's own largest fixed heap.
The checks certify eventual repetition as above; they do not
prove the conjecture for arbitrarily many heaps.

The single-fixed-heap period need not itself be a universal period
for all boards under that heap limit. For example,
\[
 g(2,3,n)=
 \begin{cases}
 0,&n\text{ odd},\\
 2,&n\equiv0\pmod4,\\
 3,&n\equiv2\pmod4.
 \end{cases}
\]
Its period is $4$, whereas $g(3,n)$ has period $2$. Both divide
the conjectured common period $2J(3)=12$.

Adding heaps need not change the least period in one direction.
The following certified example first shortens the period and
then partly reverses that reduction:

\begin{table}[H]
\centering
\begin{tabular}{@{}l r@{}}
\toprule
Fixed board $A$ & Least eventual period of $g(A,n)$\\
\midrule
$(8)$ & $1{,}680$\\
$(1,8,8,8)$ & $2$\\
$(1,4,8,8,8)$ & $12$\\
\bottomrule
\end{tabular}
\caption{Adding fixed heaps can shorten a period and then increase it.}
\label{tab:periodjoins}
\end{table}

In the last step, the smallest heap stays $1$, the largest stays
$8$, the GCD stays $1$, and the LCM stays $8$. Every divisor of
the added heap $4$ already divides $8$, yet the period rises from
$2$ to $12$. For these rows, the common shift $1{,}680$ was checked
from $n=216$ on all fixed descendants; full-cycle checks then give
the least periods. Thus unchanged heap extremes and divisors do
not guarantee an unchanged period. The conjecture asks instead
for a common period that survives arbitrary additions of heaps
of size at most $h_{\max}$.

\section{Composition: exact results, then bounds}
\label{sec:composition}

Join two boards by putting all their heaps under one shared divisor
rule, and write the result as $A\join B$. A move on a heap from $A$
must now satisfy the restrictions from $B$ too. This is different
from playing two independent games side by side: the join changes
the legal moves themselves.
Consequently, $g(A\join B)$ need not equal $g(A)\mathbin{\oplus}g(B)$,
where $\oplus$ is the bitwise exclusive-or used for independent games.

\subsection{Exact Grundy values when the join has at least two odd heaps}

The numbers $o_A$ and $o_B$ count the odd heaps in $A$ and $B$
before any halving. Joining the boards adds these counts.
If the total is at least two, its parity gives the exact SG value.

\begin{proposition}[Exact SG composition in the odd-heap case]
\label{prop:oddjoin}
If $o_A+o_B\ge2$, then
\[
 \boxed{g(A\join B)=
 \begin{cases}
 0,&o_A+o_B\text{ is even},\\
 1,&o_A+o_B\text{ is odd}.
 \end{cases}}
\]
\end{proposition}
\begin{proof}
Apply Lemma~\ref{lem:odd} to the joined board.
\end{proof}

One odd heap from each board is enough. The sizes of those heaps and
the sizes and number of even heaps do not change this value.

This also gives a replacement rule. A board $A$ with at least two
odd heaps can be replaced by $(1,1)$ if its odd count is even, or by
$(1,1,1)$ if its odd count is odd. For every future board $B$, the
original and replacement joins retain at least two odd heaps and
have the same odd-count parity. Their SG values therefore agree.

\Needspace{16\baselineskip}
\subsection{Exact winning and losing outcomes for every join}

For each nonempty board, retain its minimum 2-adic depth and
the number of heaps at that depth. These are $(k_A,r_A)$ for
$A$ and $(k_B,r_B)$ for $B$. If one board reaches an odd heap
first under common halving, it determines where halving stops
for the join. If both reach odd heaps together, their counts add.

\begin{center}
\begin{tabularx}{\linewidth}{@{}L{0.30\linewidth}Y@{}}
\toprule
Quantity and meaning & Expression and reason \\
\midrule
$k$\newline
Minimum 2-adic depth of the join
  & $k=\min(k_A,k_B)$. Common halving stops as soon as
    either board has an odd heap. \\
\addlinespace[9pt]
$r$\newline
Number of joined heaps at depth $k$
  & Only heaps at the smaller depth are counted; if the
    depths agree, both boards contribute:
    \[
      r=\begin{cases}
      r_A,&k_A<k_B,\\
      r_B,&k_B<k_A,\\
      r_A+r_B,&k_A=k_B.
      \end{cases}
    \] \\
\bottomrule
\end{tabularx}
\end{center}

Combining these quantities with Morisawa's outcome criterion gives
\begin{equation}
 \boxed{g(A\join B)=0\quad\Longleftrightarrow\quad r\text{ is even}.}
 \label{eq:outcomejoin}
\end{equation}
An odd count means the join is winning.

For any number of boards, take the smallest depth and add the counts
from boards at that depth. This does not depend on the order or
grouping of the joins, so the composition rule is commutative and
associative. Empty boards can be omitted.

Only the parity of $r$ is needed for outcomes. Two boards with the
same minimum depth and the same parity of the count at that depth can replace one another
in every future join without changing its outcome. A board joined to
either one has smaller, equal, or larger depth; in each case,
the count rule gives the same final parity.

Thus full SG composition is exact in the odd-heap case, and outcome
composition is exact for all boards. For the remaining positive SG
values, we use the bounds already proved.

\subsection{The SG bounds compose too}

The SG bound uses two inputs: the smallest heap $h_{\min}$
and the minimum 2-adic depth $k$. Joining boards takes the
minimum of each input. We already have $k$ from the outcome
calculation; adding the smallest heap gives the bound.

\begin{table}[H]
\centering
\begin{tabularx}{\linewidth}{@{}L{0.30\linewidth}Y@{}}
\toprule
Quantity and meaning & Expression and reason \\
\midrule
$h_{\min}$\newline
Smallest heap of the join
  & $h_{\min}=\min\bigl(h_{\min}(A),h_{\min}(B)\bigr)$.
    The smallest heap in the combined board is the smaller
    of the two original smallest heaps. \\
\addlinespace[9pt]
$U(h_{\min},k)$\newline
SG ceiling for the join
  & Substituting the joined smallest heap and minimum depth
    into the bound from Section~\ref{sec:values} gives
    \[
      U(h_{\min},k)=2h_{\min}-\frac{h_{\min}}{2^k}+k+1.
    \]
    The count $r$ then determines whether the value is
    exactly zero or the ceiling can be tightened. \\
\bottomrule
\end{tabularx}
\caption{The quantities used to bound the SG value of a join.}
\label{tab:joinsummary}
\end{table}

The ceiling $U(h_{\min},k)$ increases with $h_{\min}$,
since its coefficient $2-2^{-k}$ is positive. It also increases with
$k$, because
\[
 U(h_{\min},k+1)-U(h_{\min},k)=h_{\min}/2^{k+1}+1>0.
\]
Taking the minimum of each input therefore cannot increase the ceiling.

\Needspace{13\baselineskip}
\begin{proposition}[A bound preserved by joining]
\label{prop:boundjoin}
For nonempty boards $A,B$, let $(h_{\min},k,r)$ be their joined summary.
Then
\begin{equation}
 \boxed{g(A\join B)\le U(h_{\min},k)
 \le\min\{U(h_{\min}(A),k_A),U(h_{\min}(B),k_B)\}.}
 \label{eq:boundjoin}
\end{equation}
If $r$ is even, the exact value is $0$. If $r\ge3$ is odd, the first
ceiling improves to $U(h_{\min},k)-h_{\min}/2^k$.
\end{proposition}
\begin{proof}
Apply Corollary~\ref{cor:closed} and \eqref{eq:outcomejoin}
to the joined board. The comparison with the two original ceilings follows from
the increase of $U$ in each input.
\end{proof}

The smallest heap and smallest depth may come from different boards,
so their combination can improve on both original ceilings. For
example, let $A=(4,16)$ and $B=(6,16)$. Their ceilings are $10$ and
$11$. The join takes $h_{\min}=4$ from $A$ and $k=1$ from $B$, giving $8$.
Adding heaps of size $2$ then illustrates the role of the count.

\begin{table}[H]
\centering\small
\begin{tabular}{@{}lcccc@{}}
\toprule
Board & $h_{\min}$ & $k$ & $r$ & SG bound or exact value\\
\midrule
$A=(4,16)$ & $4$ & $2$ & $1$ & $10$\\
$B=(6,16)$ & $6$ & $1$ & $1$ & $11$\\
$A\join B$ & $4$ & $1$ & $1$ & $8$\\
$A\join B\join(2,2)$ & $2$ & $1$ & $3$ & $5-1=4$\\
$A\join B\join(2,2,2)$ & $2$ & $1$ & $4$ & $0$ (exact)\\
\bottomrule
\end{tabular}
\caption{The positive entries are upper bounds. The last entry is
the exact losing value.}
\label{tab:joinexample}
\end{table}

In the row with $r=3$, Corollary~\ref{cor:closed} improves
$U(2,1)=5$ by $h_{\min}/2^k=2/2=1$, giving the bound $4$.

Finally, start with any fixed nonempty board $A$ and keep adding
boards. Repeated application of \eqref{eq:boundjoin} gives
\begin{equation}
 \boxed{g(A\join B_1\join\cdots\join B_t)\le U(h_{\min}(A),k_A)
 \quad\text{for every }t\ge0.}
 \label{eq:forever}
\end{equation}
The actual value may rise or fall as heaps are added. The ceiling
provided by the original board remains valid, regardless of how
many additional heaps appear or how large they are.

\subsection{Composing period bounds}

For periodicity, $A$ and $B$ are fixed boards and $n$ is one
varying heap attached to their join. The explicit common period
uses the largest fixed heap and the total fixed stones. Joining
boards takes the larger heap size and adds the stone counts.

\begin{center}
\begin{tabularx}{\linewidth}{@{}L{0.30\linewidth}Y@{}}
\toprule
Quantity and meaning & Expression and reason \\
\midrule
$h_{\max}$\newline
Largest fixed heap of the join
  & The largest heap comes from one of the two boards:
    \[
      h_{\max}=\max(h_{\max}(A),h_{\max}(B)).
    \]
    This same size limit remains valid in every fixed
    descendant of the join. \\
\addlinespace[9pt]
$s(A\join B)$\newline
Joined number of allowed layers
  & Each $s$ is the total fixed stones minus one. Hence
    \[
      s(A\join B)=s(A)+s(B)+1.
    \]
    The extra $1$ corrects for subtracting one in each
    input count. \\
\bottomrule
\end{tabularx}
\end{center}

Substituting these joined quantities gives a common period
without calculating any successor periods.

\begin{corollary}[A period bound for joined boards]
\label{cor:periodjoin}
For nonempty fixed boards $A,B$, use their joined largest heap
$h_{\max}$. Then
\begin{equation}
 \boxed{\omega_{A\join B}\mid
 T_0(h_{\max})K(h_{\max})^{s(A)+s(B)+1}.}
 \label{eq:periodjoin}
\end{equation}
The expression on the right is a guaranteed eventual period.
\end{corollary}
\begin{proof}
Apply Proposition~\ref{prop:explicitperiod} to $A\join B$
and use the two quantities in the table.
\end{proof}

For several boards, take the largest fixed heap and add all fixed
stones. These operations do not depend on the grouping or order
of the joins. If heaps no larger than $h_{\max}$ are added, the
possible prime collection remains unchanged. The guaranteed
period can still grow through its exponent $s$.

Under Conjecture~\eqref{eq:maxconjecture}, the join would instead
have common period $2J(h_{\max})$, independent of those stone
counts. This would compose a period guarantee from the largest
fixed heap alone. It would not imply that the actual least period
always increases or always decreases when boards are joined.

\section{Conclusions and open questions}
\label{sec:questions}

The smallest nonempty heap bounds the current SG value by
$C(h_{\min})=2h_{\min}+\lfloor\log_2h_{\min}\rfloor$,
independently of all other heap sizes and counts. The largest
initial heap supplies a bound for every descendant. Joining
boards preserves and can tighten the SG ceiling; minimum
2-adic depth and its count parity determine every joined outcome.

For periods, the smallest fixed heap controls the local terms.
The largest original fixed heap keeps one set of bounds valid
throughout the successor chains. The resulting common period is
\[
 T_0(h_{\max})K(h_{\max})^{s(A)},
 \qquad s(A)=\sum_{h\in A}h-1.
\]
LCM combines all successors at one layer without a factor for
their number. Moving up a layer may add another multiplier,
which is why the exponent still depends on the total fixed stones.
The possible prime divisors nevertheless belong to one finite
set determined by $h_{\max}$ alone.

The two principal conjectures are:

\begin{center}
\begin{tabularx}{\linewidth}{@{}L{0.30\linewidth}Y@{}}
\toprule
Conjecture and meaning & Precise claim and significance \\
\midrule
Sharper SG bound
 & For every nonempty position $P$,
   \[
     \boxed{g(P)\le2h_{\min}(P).}
   \]
   This would remove the remaining logarithmic term. \\
\addlinespace[9pt]
Common period independent of heap count
 & For every nonempty fixed board $A$,
   \[
     \boxed{\omega_A\mid2J(h_{\max}(A)).}
   \]
   Here $J(h_{\max})=\operatorname{lcm}(1,\ldots,h_{\max})$
   and $h_{\max}(A)$ excludes the varying heap $n$.
   This would bound both the primes and their powers using
   the largest fixed heap alone. \\
\bottomrule
\end{tabularx}
\end{center}

The period conjecture holds for every board in the certified
finite ranges reported above. A proof must still explain why
arbitrarily many fixed heaps cannot keep adding new powers
through recursive layers. Removing successor values that do not
affect the mex, and understanding repeated copies of a heap,
may help control those layers. The size at which eventual
repetition begins may still depend on the whole fixed board.

\end{document}